\documentclass[letterpaper, 10 pt, conference]{ieeeconf} 
\IEEEoverridecommandlockouts  

\usepackage{amsmath}
\usepackage{amssymb}
\usepackage{amsthm}
\usepackage{mathtools}
\usepackage{derivative}
\usepackage{xcolor}
\usepackage{bm}

\usepackage[noadjust]{cite}
\usepackage{url}

\usepackage{graphicx}
\usepackage[export]{adjustbox} 
\graphicspath{{figures/}}

\newtheorem{theorem}{Theorem}
\newtheorem{lemma}{Lemma}

\newtheorem{corollary}{Corollary}

\theoremstyle{definition}
\newtheorem{definition}{Definition}

\newcommand{\argmin}{\operatornamewithlimits{arg\,min}}

\newcommand{\R}{\mathbb{R}}

\newcommand{\mc}[1]{\mathcal{#1}}
\newcommand{\T}{^\top}

\newcommand{\Kinf}{\mathcal{K}_{\infty}}

\newcommand{\map}[3]{#1\,:\,#2\rightarrow #3}
\newcommand{\Kinfe}{\mathcal{K}_{\infty}^e}
\newcommand{\rank}{\operatorname{rank}}
\newcommand{\Int}{\operatorname{Int}}
\newcommand{\diag}{\operatorname{diag}}

\newcommand{\bzero}{\mathbf{0}}

\renewcommand{\bf}{\mathbf{f}} 
\newcommand{\bg}{\mathbf{g}}

\newcommand{\bk}{\mathbf{k}}

\newcommand{\bp}{\mathbf{p}}
\newcommand{\bq}{\mathbf{q}}

\newcommand{\bu}{\mathbf{u}}
\newcommand{\bv}{\mathbf{v}}

\newcommand{\bx}{\mathbf{x}}
\newcommand{\by}{\mathbf{y}}

\newcommand{\bA}{\mathbf{A}}
\newcommand{\bB}{\mathbf{B}}
\newcommand{\bC}{\mathbf{C}}
\newcommand{\bD}{\mathbf{D}}

\newcommand{\bG}{\mathbf{G}}

\newcommand{\boldeta}{\bm{\eta}}

\newcommand{\bzeta}{\bm{\zeta}}
\newcommand{\bomega}{\bm{\omega}}

\newcommand{\Q}{\mathcal{Q}}

\usepackage{xcolor}
\definecolor{myblue}{RGB}{49, 114, 174}
\definecolor{myred}{rgb}{0.796, 0.235, 0.2}
\definecolor{mygreen}{rgb}{0.22, 0.596, 0.149}
\definecolor{mypurple}{rgb}{0.584,0.345,0.698}

\usepackage{blindtext}
\usepackage{fancyhdr}

\fancypagestyle{FirstPage}{
\lfoot{\footnotesize \copyright 2024 IEEE.  Personal use of this material is permitted.  Permission from IEEE must be obtained for all other uses, in any current or future media, including reprinting/republishing this material for advertising or promotional purposes, creating new collective works, for resale or redistribution to servers or lists, or reuse of any copyrighted component of this work in other works.} 
}

\title{\textbf{Constructive Safety-Critical Control: Synthesizing Control Barrier Functions for Partially Feedback Linearizable Systems}}

\author{Max H. Cohen, Ryan K. Cosner, and Aaron D. Ames %
\thanks{This research was supported by BP
and NSF CPS Award \#1932091.}
\thanks{The authors are with the Department of Mechanical and Civil Engineering, California Institute of Technology, Pasadena, CA \texttt{\{maxcohen,rkcosner,ames\}@caltech.edu}.}
}

\begin{document}
\maketitle

\begin{abstract}
    Certifying the safety of nonlinear systems, through the lens of set invariance and control barrier functions (CBFs), offers a powerful method for controller synthesis, provided a CBF can be constructed.
    This paper draws connections between partial feedback linearization and CBF synthesis.
    We illustrate that when a control affine system is input-output linearizable with respect to a smooth output function, then, under mild regularity conditions, one may 
    extend any safety constraint defined on the output to a CBF for the full-order dynamics.
    These more general results are specialized to robotic systems where the conditions required to synthesize CBFs simplify. The CBFs constructed from our approach are applied and verified in simulation and hardware experiments on a quadrotor.
\end{abstract}

\section{Introduction}
\label{sec:introduction}
\thispagestyle{FirstPage}
Safety has emerged as a fundamental requirement for modern control systems.  With safety framed as set invariance, control barrier functions (CBFs) have become a popular tool for designing controllers that endow systems with safety guarantees \cite{AmesTAC17}. 
Given a CBF, one may 
construct a controller enforcing set invariance using convex optimization \cite{AmesTAC17} or smooth universal formulas \cite{CohenLCSS23,KrsticTAC24}. Synthesizing a CBF-based controller, however, first requires constructing a valid CBF -- a task often cited as the greatest challenge of CBF-based approaches. For low-dimensional systems, computational techniques such as sum-of-squares programming \cite{xu2017correctness,HovakimyanLCSS23} and Hamilton-Jacobi reachability \cite{HerbertCDC21,wabersich2023data} often provide a viable pathway to address such challenges. Alternatively, one may construct hand-crafted CBFs for specific systems \cite{SreenathACC16-quad}. Yet each of these approaches tends to scale poorly with the dimension and complexity of the system.

A popular approach to constructing CBFs for complex high-dimensional systems is via backstepping \cite{AndrewCDC22,CohenARC24,AbelTAC23}. This approach effectively extends a safety constraint that is not controlled invariant to a CBF defining a control invariant set (the zero superlevel set) contained within the original safety constraint. This is accomplished by imposing a particular layered structure on the dynamics, defining smooth virtual CBF-based controllers \cite{CohenLCSS23} at each layer, and then ``backstepping'' through such controllers to compute a CBF for the overall system. The main limitations of this approach are the structural requirements of the dynamics and the controllability assumptions imposed on each layer, which precludes the direct application of such ideas to underactuated systems. 

\begin{figure}[t]
    \centering
    \includegraphics[width=0.45\textwidth]{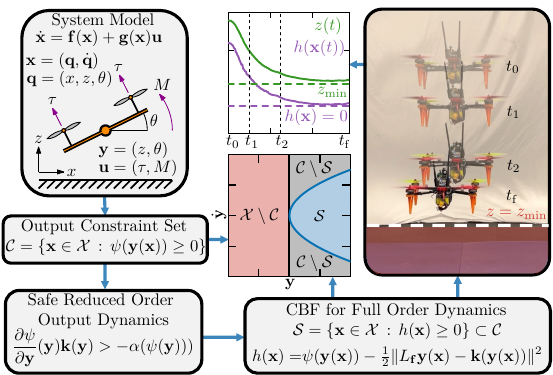}
    \vspace{-0.3cm}
    \caption{We present a methodology to systematically generate control barrier functions for high-dimensional underactuated systems from inequality constraints on the system's output. A video of an experimental demonstration of our approach can be found at \texttt{https://youtu.be/GYvQjcojLIQ}.}
    \vspace{-0.75cm}
    \label{fig:master_figure}
\end{figure}

Similar to backstepping, high order CBFs (HOCBFs) \cite{SreenathACC16,XuAutomatica18,WeiTAC22,TanTAC22} extend a safety constraint to a barrier-like function that may be used to enforce forward invariance of a safe set contained within the constraint set.
In contrast to backstepping, HOCBFs place no structural requirements on the dynamics other than that they are control affine. Inspired by input-output linearization \cite{Isidori}, such approaches treat the safety constraint as an output, differentiate this output until the input appears, and then impose CBF-like conditions on the highest derivative of the output. However, by treating the safety constraint as an output, these approaches implicitly require such a constraint to have a uniform relative degree on the safe set, which is restrictive in the context of CBFs \cite{jankovic2018robust}. As noted in \cite{CohenARC24, TanTAC22, Cohen}, even simple safety constraints may not have a uniform relative degree and, in such a situation, the functions constructed following the approaches in \cite{SreenathACC16,XuAutomatica18,WeiTAC22,TanTAC22} may not meet the criteria of a HOCBF. Similar limitations arise when using such a methodology to construct CBFs, rather than HOCBFs, from a high relative degree safety constraint \cite{BreedenAutomatica23}.

In this paper, we demonstrate how techniques from feedback linearization \cite{Isidori} facilitate the construction of CBFs, with an emphasis on the application of such ideas to underactuated robotic systems. 
Instead of treating safety constraints directly as outputs, as in prior works such as \cite{AbelTAC23,SreenathACC16,XuAutomatica18,WeiTAC22,BreedenAutomatica23}, we define the states relevant to the safety constraint as outputs. Leveraging the structural properties of the resulting output dynamics, we employ methods from \cite{AndrewCDC22,CohenARC24} to construct CBFs for the full-order dynamics, thereby relaxing the restrictive uniform relative degree requirements found in existing high relative degree CBF frameworks.
Specifically, we establish that when a nonlinear control system is input-output linearizable with respect to a smooth output function, then, under mild regularity conditions, one may extend any smooth inequality constraint on the output to a CBF for the full-order system
(Sec. \ref{sec:feedback-lin}). We illustrate the utility of these results by specializing them to robotic systems where the conditions required to construct CBFs simplify  (Sec. \ref{sec:underactuated}). The benefits of our approach are highlighted through both simulations and hardware demonstrations (Sec. \ref{sec:sims}). 

To summarize, the contributions of this paper are twofold:
\begin{itemize}
    \item We present a framework for constructing CBFs for high-dimensional and underactuated systems inspired by the methods developed in \cite{AndrewCDC22}.  In contrast to \cite{AndrewCDC22}, we establish the existence of a smooth controller required for the initial step in the CBF backstepping procedure. Furthermore, we formally characterize the properties required of the system output, safety constraint, and candidate CBF to ensure the applicability of the techniques described in \cite{AndrewCDC22} to general control affine systems.
    \item We present extensive numerical examples illustrating the design of CBFs for various underactuated robotic systems and apply the developed theory on a quadrotor, which constitutes the first demonstration of CBF backstepping on hardware.
\end{itemize}

\noindent {\bfseries Notation.}
Define $\partial \mc{S}$ and $\Int(\mc{S})$ as the boundary and interior of a set $\mc{S}$. A continuous function $\alpha\,:\,\R\rightarrow\R$ is said to be an extended class $\Kinf$ function ($\alpha\in\Kinfe$) if i) $\alpha(0)=0$, ii) $\alpha$ is strictly increasing, iii) $\lim_{r\rightarrow\pm\infty}\alpha(r)=\pm\infty$. For smooth functions
$\by\,:\,\R^n\rightarrow\R^m$ and 
$\bg\,:\,\R^n\rightarrow\R^{n\times m}$, we define 
$L_{\bg}\by(\bx)\coloneqq \pdv{\by}{\bx}(\bx)\bg(\bx)$ as the Lie derivative of $\by$ along 
$\bg$
with higher order Lie derivatives defined as in \cite{Isidori}.

\section{Preliminaries and Problem Formulation}\label{sec:prelim}
Consider a nonlinear control affine system:
\begin{equation}\label{eq:control-affine-dyn}
    \dot{\bx} = \bf(\bx) + \bg(\bx)\bu,
\end{equation}
with state $\bx\in \mc{X}\subseteq\R^n$ and control input $\bu\in \R^m$, where $\map{\bf}{\mc{X}}{\R^n}$ 
and $\map{\bg}{\mc{X}}{\R^{n\times m}}$ 
are smooth (differentiable as many times as necessary) on the open and connected set $\mathcal{X}$. By taking $\bu=\bk(\bx)$ with $\map{\bk}{\mc{X}}{\R^m}$ locally Lipschitz
we obtain the closed-loop system $\dot{\bx} = \bf(\bx) + \bg(\bx)\bk(\bx)$,
which, for each initial condition $\bx_0\in \mc{X}$, produces a continuously differentiable trajectory $\map{\bx}{I(\bx_0)}{\mc{X}}$ defined on a maximal interval of existence $I(\bx_0)\subseteq\R_{\geq0}$. A set $\mc{S}\subset \mc{X}$ is said to be \emph{forward invariant} for the closed-loop system if, for each $\bx_0\in \mc{S}$, the resulting trajectory satisfies $\bx(t)\in \mc{S}$ for all $t\in I(\bx_0)$. 
A popular approach to designing controllers enforcing 
forward invariance
is through CBFs.

\begin{definition}[\cite{AmesTAC17}]
    A continuously differentiable function $h\,:\,\mc{X}\rightarrow\R$ defining a set $\mc{S}\subset \mc{X}$ as:
    \begin{equation}\label{eq:safe-set}
        \begin{aligned}
            \mc{S} = & \{\bx\in \mc{X}\,:\,h(\bx)\geq0\}, \\ 
        \end{aligned}
    \end{equation}
    is said to be a CBF for \eqref{eq:control-affine-dyn} on $\mc{S}$ if there exists $\alpha\in \Kinfe$ and an open set $\mc{E}\subseteq \mc{X}$ satisfying $\mc{S}\subset \mc{E}$ such that for all $\bx\in \mc{E}$:
    \begin{equation}\label{eq:cbf}
        \sup_{\bu\in\R^m} \dot{h}(\bx,\bu) = \sup_{\bu\in\R^m}\left\{L_{\bf}h(\bx) + L_{\bg}h(\bx)\bu \right\}> - \alpha(h(\bx)).
    \end{equation}
    When  $h$ is a CBF, we say that $\mc{S}$ as in \eqref{eq:safe-set} is a \emph{safe set}.
\end{definition}

The existence of a CBF
implies the existence of a locally Lipschitz feedback controller $\map{\bk}{\mc{E}}{\R^m}$ enforcing the forward invariance of $\mc{S}$ \cite{AmesTAC17,jankovic2018robust}. One example of such a controller is the optimization-based safety filter:
\begin{equation}\label{eq:cbf-qp}
    \begin{aligned}
        \bk(\bx) = && \argmin_{\bu\in\R^m}\quad & \tfrac{1}{2}\|\bu - \bk_{\rm{d}}(\bx)\|^2  \\
               && \rm{s.t.} \quad & L_{\bf}h(\bx) + L_{\bg}h(\bx)\bu \geq - \alpha(h(\bx)), 
    \end{aligned}
\end{equation}
where $\map{\bk_{\rm{d}}}{\mc{X}}{\R^m}$ is a desired controller. The main objective of this paper is to systematically construct CBFs using methods 
from
feedback linearization \cite{Isidori}. Central to our approach is the notion of relative degree.
\begin{definition}[\cite{Isidori}]
    A smooth function $\map{\by}{\mc{X}}{\R^p}$ is said to have \emph{relative degree}\footnote{A vector-valued output may have different relative degrees for each of its components. For simplicity of notation, we focus on outputs whose components share the same relative degree.
    } $\gamma\in\mathbb{N}$ with respect to \eqref{eq:control-affine-dyn} on an open set $\mc{E}\subseteq\mc{X}$ if for all $\bx\in\mc{E}$:
    \begin{equation}
        \begin{aligned}
            \text{i)}\quad & L_{\bg}L_{\bf}^{i}\by(\bx) = \bzero,\quad \forall i\in\{0,\dots,\gamma-2\}, \\
            \text{ii)}\quad  & \rank(L_{\bg}L_{\bf}^{\gamma-1}\by(\bx))=p.
        \end{aligned}
    \end{equation}
\end{definition}
Let $\by$ have relative degree $\gamma$ on $\mc{E}\subseteq\mc{X}$ and define:
\begin{equation}
    \boldeta = \begin{bmatrix}
        \boldeta_1 \\ \vdots \\ \boldeta_{\gamma}
    \end{bmatrix}
    \coloneqq
    \begin{bmatrix}
        \by(\bx) \\ \vdots \\ L_{\bf}^{\gamma-1}\by(\bx)
    \end{bmatrix}
    \in\R^{p\gamma},
\end{equation}
noting that the output dynamics are given by:
\begin{equation}\label{eq:output-dyn}
    \dot{\boldeta} = \begin{bmatrix}
        \dot{\boldeta}_1 \\ \vdots \\ \dot{\boldeta}_{\gamma-1} \\ \dot{\boldeta}_{\gamma}
    \end{bmatrix}
    =
    \begin{bmatrix}
        \boldeta_2 \\ \vdots \\ L_{\bf}^{\gamma-1}\by(\bx) \\ L_{\bf}^{\gamma}\by(\bx) + L_{\bg}L_{\bf}^{\gamma-1}\by(\bx)\bu
    \end{bmatrix}.
\end{equation}
Given a smooth output $\map{\by}{\mc{X}}{\R^p}$ and a smooth inequality constraint $\map{\psi}{\R^p}{\R}$ on $\by$ defining a \emph{constraint set}:
\begin{equation}\label{eq:C}
    \mc{C} \coloneqq \{\bx\in\mc{X}\,:\, \psi(\by(\bx)) \geq 0\},
\end{equation}
that is not necessarily controlled invariant, our goal is to construct a CBF and corresponding safe set $\mc{S}\subset\mc{C}$ so that enforcing forward invariance of $\mc{S}$ leads to satisfaction of the output constraint.

\section{CBFs for Feedback Linearizable Systems}\label{sec:feedback-lin}
In this section, we establish that when \eqref{eq:control-affine-dyn} is 
partially feedback linearizable with respect to a smooth output function then, under mild regularity conditions, one may construct a CBF and corresponding safe set whose forward invariance implies satisfaction of the output constraint.
The following lemma is the starting point of our approach and outlines the regularity conditions that $\psi$ must satisfy.

\begin{lemma}\label{lemma:smooth}
    Let $\psi\,:\,\R^p\rightarrow\R$ be a smooth function defining a set
    $\mc{C}_1\subset\R^p$ as:
    \begin{equation}\label{eq:C1}
        \mc{C}_1 \coloneqq \{\by\in\R^p\,:\,\psi(\by) \geq 0\}.
    \end{equation}
    Let $\mc{D}_1\supset\mc{C}_1$ be an open set and suppose that:
    \begin{equation}\label{eq:grad-psi}
        \pdv[style-frac=\tfrac]{\psi}{\by}(\by) \neq \bzero,\quad \forall \by\in\mc{D}_1\setminus\Int(\mc{C}_1).
    \end{equation}
    Then, for any smooth $\alpha\in\Kinfe$ there exists a smooth $\bk_1\,:\,\mc{D}_1\rightarrow\R^p$ such that for all $\by\in\mc{D}_1$:
    \begin{equation}\label{eq:smooth-control}
        \pdv[style-frac=\tfrac]{\psi}{\by}(\by)\bk_1(\by) > -\alpha(\psi(\by)).
    \end{equation}
    For any $\sigma>0$, one example of such a function is given by:
    \begin{equation}\label{eq:sontag}
    \begin{aligned}
        \bk_1(\by) = & \phi\big(\alpha(\psi(\by)), \big\Vert\pdv[style-frac=\tfrac]{\psi}{\by}(\by)\big\Vert^2\big)\pdv[style-frac=\tfrac]{\psi}{\by}(\by)\T, \\
        \phi(a,b) = & \begin{cases}
            0 & b = 0 \\
            \frac{-a + \sqrt{a^2 + \sigma b^2}}{2b} & b \neq 0. \\
        \end{cases}
        \end{aligned}
    \end{equation}
\end{lemma}

\begin{proof}
    Define $a(\by)\coloneqq\alpha(\psi(\by))$, $b(\by)\coloneqq \|\pdv{\psi}{\by}(\by)\|^2$, and $\mathcal{W} \coloneqq \{(a,b)\in\R^2\,:\,a> 0\;\text{or}\; b > 0\}$.
    Using a similar argument to those in \cite{SontagSCL89,CohenLCSS23,CohenARC24}, one can show that $(a,b)\mapsto \phi(a,b)$ from \eqref{eq:sontag} is smooth on $\mc{W}$. It follows from \eqref{eq:grad-psi} that, for $\by\in\mc{D}_1$, $b(\by)=0$ only if $\by\in\Int(\mc{C}_1)$, which implies that, for $\by\in\mc{D}_1$, $b(\by)=0$ only if $a(\by)>0$. Hence, for each $\by\in\mc{D}_1$, we have $(a(\by),b(\by))\in\mc{W}$. Since $(a,b)\mapsto\phi(a,b)$ is smooth on $\mc{W}$, $\by\mapsto a(\by),b(\by)$ are smooth on $\mc{D}_1$, and $(a(\by),b(\by))\in\mc{W}$ for each $\by\in\mc{D}_1$, $\by\mapsto \phi(a(\by),b(\by))$ is smooth on $\mc{D}_1$, implying that $\by\mapsto \bk_1(\by)$ from \eqref{eq:sontag} is smooth on $\mc{D}_1$. To show \eqref{eq:smooth-control} we compute:
    \begin{equation}\label{eq:smooth-control-expression}
        \begin{aligned}
            \pdv[style-frac=\tfrac]{\psi}{\by}(\by)\bk_1(\by) = & \phi(a(\by), b(\by))\big\Vert\pdv[style-frac=\tfrac]{\psi}{\by}(\by)\big\Vert^2 \\ 
            = & \begin{cases}
                0 & b(\by)=0 \\
                \frac{-a(\by) + \sqrt{a(\by)^2 + \sigma b(\by)^2}}{2} & b(\by)\neq0.
            \end{cases}
        \end{aligned}
    \end{equation}
    Recall from \eqref{eq:grad-psi} that when $b(\by)=0$ and $\by\in\mc{D}_1$, we have $a(\by)>0$, implying that $0>-a(\by)$. Moreover, when $b(\by)\neq0$, one can verify that $\tfrac{-a(\by) + \sqrt{a(\by)^2 + \sigma b(\by)^2}}{2} > -a(\by)$ for any $\sigma>0$.
    Using these observations to lower bound \eqref{eq:smooth-control-expression} implies that \eqref{eq:smooth-control} holds, as desired.
\end{proof}
The conditions in Lemma \ref{lemma:smooth} are equivalent to the statement that $\psi$ is a CBF for a single integrator $\dot{\by}=\bu$, an arguably mild requirement.
This does not require $\bx\mapsto\psi(\by(\bx))$ to have a uniform relative degree on $\mc{C}$, which would be overly restrictive. Indeed, the gradient of relevant safety constraints may vanish at points on $\Int(\mc{C})$ \cite{TanTAC22,Cohen,CohenARC24}. Instead, we will require the output $\by$ to have a relative degree, which is less restrictive\footnote{That is, $\by$ may have a relative degree even when $\psi$ does not. A simple example illustrating this point is the double integrator with state $\bx=(x_1,x_2)\in\R^2$, output $\by(\bx)=x_1$, and constraint $\psi(\by(\bx))=1 - x_1^2$. This phenomenon is also present in the examples considered in Sec. \ref{sec:sims} and may arise when $\mc{C}_1$ from \eqref{eq:C1} is a compact set (cf. \cite[Footnote 4]{CohenARC24}).}.
To this end, let $\by$ have relative degree $\gamma\in\mathbb{N}$ on $\mc{E}\subseteq\mc{X}$ and consider the output dynamics \eqref{eq:output-dyn} of system \eqref{eq:control-affine-dyn}. The output dynamics in \eqref{eq:output-dyn} are in \emph{strict feedback form} and are thus amenable to backstepping \cite{AndrewCDC22,CohenARC24}. We will thus leverage backstepping for \eqref{eq:output-dyn} to construct a CBF, which may be used to enforce satisfaction of the original output constraint on \eqref{eq:control-affine-dyn}. 
Following 
\cite{AndrewCDC22}, we propose the CBF candidate:
\begin{equation}\label{eq:h-fbl}
    \begin{aligned}
        h(\bx) \coloneqq & \psi(\by(\bx)) - \sum_{i = 1}^{\gamma-1}\tfrac{1}{2\mu_i}\|L_{\bf}^{i}\by(\bx) - \bk_{i}(\bzeta_{i}(\bx))\|^2 \\
        = &\psi(\boldeta_1) - \sum_{i = 1}^{\gamma-1}\tfrac{1}{2\mu_i}\|\boldeta_{i+1} - \bk_{i}(\bzeta_{i})\|^2,
    \end{aligned}
\end{equation}
where $\psi$ defines $\mc{C}_1\subset\R^p$ as in \eqref{eq:C1}, $\mu_i>0$ for $i\in\{1,\dots,\gamma-1\}$, $\bzeta_j\coloneqq (\boldeta_1,\boldeta_2,\dots,\boldeta_j)\in\R^{pj}$,
$\map{\bk_1}{\mc{D}_1}{\R^p}$ is any smooth function satisfying \eqref{eq:smooth-control} for all $\boldeta_1\in\mc{D}_1\supset\mc{C}_1$ for a smooth globally Lipschitz $\alpha\in\Kinfe$ and:
\begin{equation}\label{eq:virtual-controllers}
\begin{aligned}
    \bk_{2}(\bzeta_2) \coloneqq & \dot{\bk}_{1}(\bzeta_2) + \mu_{1}\pdv{\psi}{\boldeta_1}(\boldeta_1)\T - \frac{\lambda_1}{2}(\boldeta_2 - \bk_{1}(\boldeta_1)) \\
    \bk_{i+1}(\bzeta_{i+1}) \coloneqq & \dot{\bk}_{i}(\bzeta_{i+1}) - \mu_{i}(\boldeta_i - \bk_{i-1}(\bzeta_{i-1})) \\ &- \frac{\lambda_i}{2}(\boldeta_{i+1} - \bk_{i}(\bzeta_i)),\; \forall i\in\{2,\dots,\gamma-2\}, \\
\end{aligned}
\end{equation}
where $\lambda_i>0$ for each $i\in\{1,\dots,\gamma-2\}$. The CBF candidate in \eqref{eq:h-fbl} defines a set $\mc{S}$ as in \eqref{eq:safe-set}, which satisfies $\mc{S}\subset\mc{C}$.
Before proceeding, it will be useful to define $\mc{D}\coloneqq \{\bx\in\mc{X}\,:\,\by(\bx)\in\mc{D}_1\}$, where $\mc{D}_1\subset\R^p$ is defined as in Lemma \ref{lemma:smooth}. We now illustrate that when $\by$ has a 
relative degree on $\mc{S}$ and $\psi$ satisfies \eqref{eq:grad-psi}, then \eqref{eq:h-fbl} is a CBF for \eqref{eq:control-affine-dyn}.

\begin{theorem}\label{theorem:cbf-fbl}
    Consider system \eqref{eq:control-affine-dyn} with smooth output $\map{\by}{\mc{X}}{\R^p}$, the output constraint $\map{\psi}{\R^p}{\R}$ defining a constraint set $\mc{C}\subset\mc{X}$ as in \eqref{eq:C}, and the CBF candidate $h\,:\,\mc{X}\rightarrow\R$ from \eqref{eq:h-fbl} defining a set $\mc{S}\subset\mc{C}$ as in \eqref{eq:safe-set}. Provided that $\psi$ satisfies \eqref{eq:grad-psi} on a set $\mc{D}_1\supset\mc{C}_1$, with $\mc{C}_1\subset\R^p$ as in \eqref{eq:C1}, $\by$ has relative degree $\gamma$ on a set $\mc{E}\supset\mc{S}$ satisfying $\mc{E}\subseteq\mc{D}$, and $\lambda_i \geq \ell_\alpha$ for each $i\in\{1,\dots,\gamma-2\}$, where $\ell_{\alpha}$ is a Lipschitz constant of $\alpha\in\Kinfe$ from \eqref{eq:smooth-control}, then $h$ is a CBF for \eqref{eq:control-affine-dyn} on $\mc{S}$. Moreover, any locally Lipschitz controller $\bk\,:\,\mc{E}\rightarrow\R^m$ that renders $\mc{S}$ forward invariant for the closed-loop system \eqref{eq:control-affine-dyn} ensures that $\bx(t)\in\mc{C}$ for all $t\in I(\bx_0)$.
\end{theorem}

\begin{proof}
    The proof follows a similar argument to that of \cite[Thm. 5]{AndrewCDC22}. Since $\by$ has relative degree $\gamma$ on $\mc{E}$, the matrix $L_{\bg}L_{\bf}^{\gamma-1}\by(\bx)\in\R^{p\times m}$ has rank $p$ and is thus right pseudo-invertible for each $\bx\in\mc{E}$. 
    Now, note that since $\mc{S}\subset\mc{E}\subseteq\mc{D}$ and $\boldeta_1\mapsto\bk_1(\boldeta_1)$ satisfies \eqref{eq:smooth-control} for all $\boldeta_1\in\mc{D}_1$, $\bx\mapsto\bk_1(\by(\bx))$ satisfies \eqref{eq:smooth-control} for all $\bx\in\mc{E}\subseteq\mc{D}$, where $\bk_1$ exists since $\psi$ satisfies the conditions of Lemma \ref{lemma:smooth}. It then follows that since $\lambda_i \geq \ell_\alpha$ for each $i\in\{1,\dots,\gamma-2\}$, each $\bk_i$ satisfies the same conditions as those in the proof of \cite[Thm. 5]{AndrewCDC22}, which implies that the CBF candidate $h$ in \eqref{eq:h-fbl} satisfies the same conditions as those in \cite[Sec. IV]{AndrewCDC22}. Hence, by following the same steps as in the proof of \cite[Thm. 5]{AndrewCDC22}, one may show that the smooth feedback controller:
    \begin{equation*}
        \begin{aligned}
            \bk(\bx)\coloneqq L_{\bg}L_{\bf}^{\gamma-1}\by(\bx)^{\dagger}\bigg[\dot{\bk}_{\gamma-1}(\boldeta(\bx)) - L_{\bf}^{\gamma}\by(\bx) \\ - \mu_{\gamma-1}\Big(\boldeta_{\gamma-1}(\bx)- \bk_{\gamma-2}(\bzeta_{\gamma-2}(\bx))\Big)\\
            -\frac{\lambda_{\gamma-1}}{2}\Big(\boldeta_{\gamma}(\bx) - \bk_{\gamma-1}(\bzeta_{\gamma-1}(\bx))\Big)\bigg],
        \end{aligned}
    \end{equation*}
    where $(\cdot)^{\dagger}$ denotes the right psuedo-inverse and $\lambda_{\gamma-1}\geq\ell_{\alpha}$, satisfies $\dot{h}(\bx,\bk(\bx))>-\alpha(h(\bx))$ for all $\bx\in\mc{E}$, where $\alpha$ is from \eqref{eq:smooth-control}. Thus, for all $\bx\in\mc{E}$, we have:
    \begin{equation*}
        \sup_{\bu\in\R^m}\dot{h}(\bx,\bu) \geq \dot{h}(\bx,\bk(\bx)) > -\alpha(h(\bx)),
    \end{equation*}
    which implies that $h$ is a CBF for \eqref{eq:control-affine-dyn} on $\mc{S}\subset\mc{E}$.
    Since $\mc{S}\subset\mc{C}$ any locally Lipschitz controller enforcing the forward invariance of $\mc{S}$ ensures that $\bx(t)\in\mc{C}$ for all $t\in I(\bx_0)$.
\end{proof}

Theorem 1 highlights the interplay between the output $\by$, the safety constraint $\psi$, the system's actuation capabilities, and the ability to construct CBFs. 
By ensuring that $\by$ has a relative degree on $\mc{E}\supset\mc{S}$, \eqref{eq:control-affine-dyn} may be partially transformed into a strict feedback system \eqref{eq:output-dyn} 
on $\mc{E}$, enabling the application of backstepping \cite{AndrewCDC22} to construct a CBF.
Theorem \ref{theorem:cbf-fbl} characterizes the requirements on $\psi$, $\mc{S}$, and $\by$ for such techniques to be applicable to general control affine systems \eqref{eq:control-affine-dyn}, complimenting the ideas introduced in \cite{AndrewCDC22}, which focused on systems already in strict feedback form.
While using outputs to transform a system into strict feedback form is well-established in the backstepping literature \cite{Krstic}, and has been exploited in the context of CBFs \cite{AbelTAC23} by viewing $\psi$ as an output, Theorem \ref{theorem:cbf-fbl} is, to our knowledge, the first to make the explicit connection between more general outputs and the constructions of CBFs.
As demonstrated in Sec. \ref{sec:sims}, this connection has important practical implications as it
enables the application of such ideas to a broader class of systems than those originally considered in \cite{AndrewCDC22,CohenARC24}.
Moreover, by not treating $\psi$ as an output -- a common approach in works such as \cite{AbelTAC23,SreenathACC16,XuAutomatica18,WeiTAC22,TanTAC22,BreedenAutomatica23} -- this construction overcomes the restrictive uniform relative degree requirements on $\psi$ present in most high relative degree CBF techniques.

\section{CBFs for Underactuated Robotic Systems}\label{sec:underactuated}
We now specialize the previous results to 
robotic systems with generalized coordinates $\bq\in\Q\subseteq\R^n$ 
and dynamics:
\begin{equation}\label{eq:robot-dyn}
    \bD(\bq)\ddot{\bq} + \bC(\bq,\dot{\bq})\dot{\bq} + \bG(\bq) = \bB(\bq)\bu.
\end{equation}
Here, $\dot{\bq}\in\R^n$ is the generalized velocity, $\bD(\bq)\in\R^{n\times n}$ denotes the positive definite and symmetric inertia matrix, $\bC(\bq,\dot{\bq})\in\R^{n\times n}$ denotes the Coriolis matrix, $\bG(\bq)\in\R^n$ represents gravitational and other potential effects, $\bB(\bq)\in\R^{n\times m}$ is the actuation matrix, and $\bu\in\R^m$ is the control input. 
Note that by defining $\bx\coloneqq(\bq,\dot{\bq})\in\mc{X} =  \Q\times\R^n$
we may represent \eqref{eq:robot-dyn} as in \eqref{eq:control-affine-dyn} with 
dynamics:
\begin{equation}\label{eq:robot-dyn-control-affine}
    \underbrace{
    \begin{bmatrix}
        \dot{\bq} \\ -\bD(\bq)^{-1}\left[\bC(\bq,\dot{\bq})\dot{\bq} + \bG(\bq) \right]
    \end{bmatrix}}_{\bf(\bx)}
    ,\quad
    \underbrace{
    \begin{bmatrix}
        \bzero_{n\times m} \\ \bD(\bq)^{-1} \bB(\bq)
    \end{bmatrix}}_{\bg(\bx)}.
\end{equation}
Now, consider a twice continuously differentiable output $\map{\by}{\Q}{\R^p}$, 
which is used to define an output constraint $\map{\psi}{\R^p}{\R}$ and associated output constraint set:
\begin{equation}\label{eq:S1-robotic}
    \mc{C} \coloneqq \{\bq\in\Q\,:\,\psi(\by(\bq))\geq0\},
\end{equation}
defined in the configuration space $\mc{Q}$ of \eqref{eq:robot-dyn}.
Differentiating the output $\by$ twice along the vector fields in \eqref{eq:robot-dyn-control-affine} leads to $\ddot{\by} = L_{\bf}^{2}\by(\bq,\dot{\bq}) + \pdv[style-frac=\tfrac]{\by}{\bq}(\bq)\bD(\bq)^{-1}\bB(\bq)\bu$. Importantly, we see that the the $p\times m$ ``decoupling'' matrix:
\begin{equation}\label{eq:decoupling-matrix}
    \bA(\bq) \coloneqq L_{\bg}L_{\bf}\by(\bq) =  \pdv[style-frac=\tfrac]{\by}{\bq}(\bq)\bD(\bq)^{-1}\bB(\bq),
\end{equation}
depends only on the configuration, implying that the relative degree 
depends only on the configuration. Note that when $\bB(\bq)\equiv\bB\in\R^{n\times m}$ is constant, $m\leq n$, and $\rank(\bB)=m$, the output $\by(\bq)=\bB\T\bq$ always has relative degree 2 as one can check that $\rank(\bA(\bq))=m$ for all $\bq\in\mc{Q}$.
When $\by$ has relative degree 2, the CBF candidate from \eqref{eq:h-fbl} simplifies to:
\begin{equation}\label{eq:h-underactuated}
    h(\bx) = \psi(\by(\bq)) - \tfrac{1}{2\mu}\|\pdv[style-frac=\tfrac]{\by}{\bq}(\bq)\dot{\bq} - \bk_{\psi}(\by(\bq))\|^2,
\end{equation}
where $\mu>0$ and $\map{\bk_{\psi}}{\mc{D}_1}{\R^p}$ is any continuously differentiable function satisfying \eqref{eq:smooth-control} for all $\by(\bq)\in\mc{D}_1\supset\mc{C}_1$. The following corollary illustrates that \eqref{eq:h-underactuated} is a CBF for \eqref{eq:robot-dyn-control-affine} provided $\psi$ satisfies \eqref{eq:grad-psi} and \eqref{eq:decoupling-matrix} has full row rank on a set containing $\mc{C}$. 

\begin{corollary}\label{proposition:underactuated}
    Consider system \eqref{eq:robot-dyn-control-affine} with twice continuously differentiable output $\map{\by}{\Q}{\R^p}$, the configuration constraint $\map{\psi}{\R^p}{\R}$ defining a set $\mc{C}\subset\Q$ as in \eqref{eq:S1-robotic}, and the CBF candidate $\map{h}{\mc{X}}{\R}$ as in \eqref{eq:h-underactuated} defining a set $\mc{S}\subset\mc{C}\times\R^n$ as in \eqref{eq:safe-set}. Provided that $\psi$ satisfies \eqref{eq:grad-psi} on a set $\mc{D}_1\supset\mc{C}_1$, with $\mc{C}_1\subset\R^p$ as in \eqref{eq:C1}, $\rank(\bA(\bq))=p$ for all $\bq\in\mc{E}_1\supset\mc{C}$ with $\mc{E}\coloneqq\mc{E}_1\times\R^n\subseteq\mc{D}$, then $h$ is a CBF for \eqref{eq:robot-dyn-control-affine}. Moreover, any locally Lipschitz controller $\bk\,:\,\mc{E}\rightarrow\R^m$ that renders $\mc{S}$ forward invariant for the closed-loop system \eqref{eq:robot-dyn-control-affine} ensures that $\bq(t)\in\mc{C}$ for all $t\in I(\bx_0)$.
\end{corollary}
\begin{proof}
    As $\rank(\bA(\bq))=p$ for all $\bq\in\mc{E}_1$, $\by$ has relative degree 2 on $\mc{E}$ and since $\mc{S}\subset\mc{C}\times\R^n$ and $\mc{C}\subset\mc{E}_1$, we have $\mc{S}\subset\mc{E}$. Finally, since $\mc{E}\subseteq\mc{D}$ the conditions of Theorem \ref{theorem:cbf-fbl} hold, implying that $h$ as in \eqref{eq:h-underactuated} is a CBF for \eqref{eq:robot-dyn-control-affine} on $\mc{S}$.
\end{proof}

Focusing on robotic systems \eqref{eq:robot-dyn}, rather than general control affine systems \eqref{eq:control-affine-dyn}, offers various benefits due to the structural properties of \eqref{eq:robot-dyn}. In particular, the relative degree of $\by$ depends only on $\bq$, implying that the relative degree can be verified over a lower-dimensional space.
This often allows one to restrict the constraint set so that $\by$ has relative degree 2 on $\mc{C}$ by construction -- a procedure exemplified in Sec. \ref{sec:quad}.
Moreover, when $\by$ has relative degree 2, the resulting CBF \eqref{eq:h-underactuated} may be defined with a general $\alpha\in\Kinfe$, rather than a smooth globally Lipschitz one as in \eqref{eq:h-fbl}. This formulation also does not require defining the $\lambda_i$ parameters in \eqref{eq:h-fbl}.

\section{Simulations and Hardware Experiments}\label{sec:sims}

\begin{figure*}
    \centering
    \includegraphics{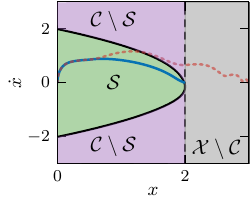}
    \hfill
    \includegraphics{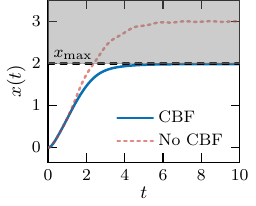}
    \hfill
    \includegraphics{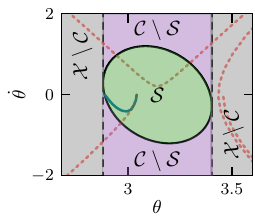}
    \hfill
    \includegraphics{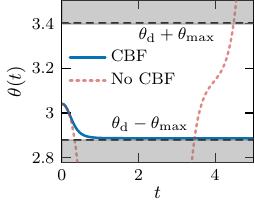}
    \vspace{-0.45cm}
    \caption{Safe sets and simulations results of the pendulum on a cart with a CBF placed on the cart's position (left) and the pendulum's angle (right). In each plot the dashed black lines denote the boundary of the constraint set and the gray regions denote the states where the constraint is violated. In the first and third plots, the green regions correspond to $\mc{S}$ and the purple region to $\mc{C}\setminus\mc{S}$.}\label{fig:cartpole}
    \vspace{-0.5cm}
\end{figure*}

\subsection{Pendulum on a Cart}
We illustrate\footnote{Code and further details of our implementation available at \url{https://github.com/maxhcohen/ReducedOrderModelCBFs.jl}} the ideas presented herein using a canonical example of an underactuated robotic system, the pendulum on a cart. The configuration $\bq=(x,\theta)\in\R\times\mathbb{S}^1\eqqcolon\Q$ consists of the planar position of the cart $x\in\R$ and the angle of a pendulum mounted on the cart $\theta\in\mathbb{S}^1$ with dynamics:
\begin{equation*}
    \begin{aligned}
        \bD(\bq) = &
        \begin{bmatrix}
            m_{c} + m_{p} & m_{p}l\cos\theta \\ m_{p}l\cos\theta & m_{p}l^2
        \end{bmatrix}, & 
        \bG(\bq) = & \begin{bmatrix}
            0 \\ m_{p}gl\sin\theta
        \end{bmatrix}, \\
        \bC(\bq,\dot{\bq}) = & \begin{bmatrix}
            0 & -m_{p}l\dot{\theta}\sin\theta \\ 0 & 0
        \end{bmatrix}, &
        \bB(\bq) = & \begin{bmatrix}
            1 \\ 0
        \end{bmatrix},
    \end{aligned}
\end{equation*}
where $m_c,m_p\in\R_{>0}$ denote the mass of the cart and pendulum, respectively, $l\in\R_{>0}$ denotes the length of the pendulum, and $g\in\R_{>0}$ denotes the acceleration due to gravity. We now demonstrate how the choice of output affects the ability to construct CBFs. Let $\by(\bq)=x$ so that our safety constraint depends only on the position of the cart. In this case, the decoupling matrix is $\bA(\bq) =\frac{m_{p}l^2}{\det\bD(\bq)}$, which has rank 1 for all $\bq\in \Q$ implying that any function of the form \eqref{eq:h-underactuated} with $\by(\bq)=x$ and $\psi$ satisfying \eqref{eq:grad-psi} is a CBF for this system. On the other hand, when $\by(\bq)=\theta$, we have $\bA(\bq) = \frac{-m_{p}l\cos\theta}{\det\bD(\bq)},$ which has rank 1 so long as $\cos\theta\neq0$. Hence, any function of the form \eqref{eq:h-underactuated} with $\psi$ satisfying \eqref{eq:grad-psi} is a CBF for this system provided that the constraint set $\mc{C}$ does not contain points such that $\cos\theta=0$. Using these observations, we construct two CBFs for the two different outputs using the configuration constraints $\psi(\by(\bq))=x_{\max} - x$ and $\psi(\by(\bq))=\theta_{\max}^2 - (\theta_{\mathrm{d}}- \theta)^2$, respectively, which require the position of the cart to remain less than $x_{\max}$ and requires the angle of the pendulum to satisfy $|\theta - \theta_{\mathrm{d}}|\leq \theta_{\max}$ with $\theta_{d}\in\mathbb{S}^1$ a desired angle of the pendulum. The gradients of each constraint are given by $\pdv{\psi}{\by}(x)=-1$ and $\pdv{\psi}{\by}(\theta)=-2(\theta_{\mathrm{d}}- \theta)$, respectively, which satisfy \eqref{eq:grad-psi} for $\mc{D}_1=\R$. These constraints are used to construct $\bk_{\psi}$ satisfying \eqref{eq:smooth-control} using \eqref{eq:sontag} and then used to construct CBFs as in \eqref{eq:h-underactuated} whose corresponding safe sets are illustrated in Fig. \ref{fig:cartpole}. These CBFs are subsequently used to construct controllers as in \eqref{eq:cbf-qp} that filter a nominal controller that attempts to drive the cart to a position beyond $x_{\max}$ and a nominal controller that applies no input, respectively. The results of applying such controllers are illustrated in Fig. \ref{fig:cartpole}. As guaranteed by Corollary \ref{proposition:underactuated}, such a controller ensures forward invariance of $\mc{S}$ and satisfaction of each constraint.

\subsection{Planar Quadrotor}\label{sec:quad}
We now apply our approach to a planar quadrotor with configuration $\bq=(x,z,\theta)\in\Q=\R^2\times\mathbb{S}^1$ consisting of the horizontal $x$ and vertical $z$ position of the quadrotor and the orientation $\theta$ of the quadrotor with respect to the horizontal plane.
The dynamics are in the form of \eqref{eq:robot-dyn} with
\cite{SreenathACC16-quad}:
\begin{equation*}
    \underbrace{
    \begin{bmatrix}
        m & 0 & 0 \\ 0 & m & 0 \\ 0 & 0 & I
    \end{bmatrix}}_{\bD(\bq)}
    \underbrace{
    \begin{bmatrix}
        \ddot{x} \\ \ddot{z} \\ \ddot{\theta}
    \end{bmatrix}}_{\ddot{\bq}}
    +
    \underbrace{
    \begin{bmatrix}
        0 \\ mg \\ 0
    \end{bmatrix}}_{\bG(\bq)}
    =
    \underbrace{
    \begin{bmatrix}
        \sin\theta & 0 \\ \cos\theta & 0 \\ 0 & -1
    \end{bmatrix}}_{\bB(\bq)}
    \underbrace{
    \begin{bmatrix}
        F \\ M
    \end{bmatrix}}_{\bu},
\end{equation*} 
where $\bC(\bq,\dot{\bq})=\bzero$, $m,I\in\R_{>0}$ are the mass and inertia, $g\in\R_{>0}$ is the acceleration due to gravity, and $F\in\R$ and $M\in\R$ are the thrust and moment applied by the propellers. Our objective is to design a controller that keeps the quadrotor's height above a specified value, which is captured by the output $\by(\bq)=z$ and the safety constraint $\psi(\by(\bq)) = z - z_{\min}$. To check if this constraint leads to a CBF via Corollary \ref{proposition:underactuated}, we first compute $\pdv{\psi}{\by}(z) = 1$ to find that our safety constraint $\psi$ satisfies \eqref{eq:grad-psi} for $\mc{D}_1=\R$. We then compute \eqref{eq:decoupling-matrix} to find that $\bA(\bq) = \frac{\cos\theta}{m}$ has rank 1 provided $\cos\theta\neq0$. However, since the configuration constraint places no limits on $\theta$, such points are contained in $\mc{C}$ and, consequently, this $\psi$ does not lead to a valid CBF. That $\bA$ drops rank on $\mc{C}$ indicates that these particular choices of $\by$ and $\psi$ are not compatible with the system's actuation capabilities, and must be refined so that $\bA$ has full rank on $\mc{C}$.
To this end, we modify our output function to $\by(\bq)=(z,\theta)$ and our constraint function to $\psi(\by(\bq)) = 1 - \tfrac{(z - z_{c})^2}{(z_c-z_{\min})^2} - \tfrac{\theta^2}{\theta_{\max}^2},$ which defines an ellipse in the $(z,\theta)$ space with center $(z_c,0)$, width $2(z_c-z_{\min})$, and height $2\theta_{\max}$. This constraint ensures that $z \geq z_{\min}$ and $|\theta|\leq \theta_{\max}$ whenever $\psi(\by(\bq)) \geq 0$. 
To check if $\psi$ yields a CBF, we compute $\pdv{\psi}{\by}(z,\theta)=[\tfrac{-2(z - z_{c})}{(z_c-z_{\min})^2}\;-2\theta]$, which satisfies \eqref{eq:grad-psi} for $\mc{D}_1=\R^2$, and $\bA(\bq)=\diag(\tfrac{\cos\theta}{m},-\tfrac{1}{I})$, where $\diag(\cdot)$ creates a diagonal matrix, which has rank 2 provided $\cos\theta\neq0$. By taking $\theta_{\max} < \tfrac{\pi}{2}$ we have the existence of a set $\mc{E}_1\supset\mc{C}$ such that $\rank(\bA(\bq))=2$ for all $\bq\in\mc{E}_1$, implying that this choice of $\by$ and $\psi$ leads to a CBF by Corollary \ref{proposition:underactuated}. Constructing a CBF using this constraint in \eqref{eq:h-underactuated}, where $\bk_{\psi}$ is from \eqref{eq:sontag}, and applying the resulting controller \eqref{eq:cbf-qp} to the system, where $\bk_{\rm{d}}$ attempts to stabilize the system to $x=0$, leads to the results in Fig. \ref{fig:quad_falling} and Fig. \ref{fig:planar_quad_plots}, where the quadrotor falls from its initial state to a height of $z_{\min}$ and maintains its height there for all time. 
The process outlined above emphasizes that designing a CBF requires carefully selecting the output $\by$ and constraint $\psi$ to ensure compatibility with the system's actuation capabilities. Importantly, the results in Sec. \ref{sec:feedback-lin} and \ref{sec:underactuated} guide this selection by providing verifiable conditions on $\by$ and $\psi$, which enables the automatic construction of a CBF when satisfied.

\begin{figure}
    \centering
    \includegraphics{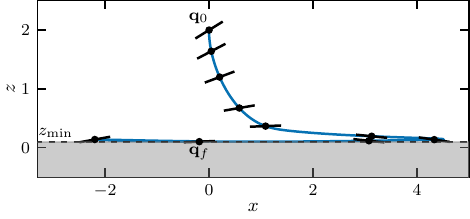}
    \vspace{-0.45cm}
    \caption{Simulations of the planar quadrotor whose CBF ensures that $z\geq z_{\min}$ where $\bq_0$ and $\bq_f$ denote the initial and final position of the quadrotor.}\label{fig:quad_falling}
    \vspace{-0.4cm}
\end{figure}

\begin{figure}
    \centering
    \includegraphics{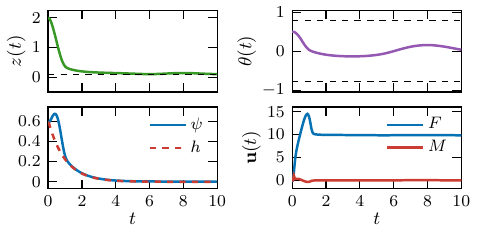}
    \vspace{-0.45cm}
    \caption{Simulation results illustrating the evolution of the planar quadrotor's height (top left), orientation (top right), the safety constraint and CBF (bottom left), and control inputs (bottom right). In the top plots, the dashed lines denote the limits imposed on $z$ and $\theta$ by the safety constraint.}
    \label{fig:planar_quad_plots}
    \vspace{-0.4cm}
\end{figure}

\subsection{Quadrotor Hardware Experiments}\label{sec:experiments}
We now extend the preceding example to a 3D quadrotor and illustrate the efficacy of our approach on hardware. The hardware platform is described in  \cite{CosnerICRA24} and is modeled as 
a control affine system \eqref{eq:control-affine-dyn}
with state $\bx=(\bp,\bq,\bv)\in\R^3\times \rm{SO}(3)\times\R^3$ representing the position $\bp$, orientation $\bq$ (represented as a quaternion), and velocity $\bv$, and control input $\bu=(\bomega,\tau)\in\mathfrak{so}(3)\times\R$, where $\bomega$ is the angular rate and $\tau$ is the thrust. A full expression of the dynamics can be found in \cite{CosnerICRA24}. Our control objective is to keep the quadrotor's height above $z_{\min}$, where $\bp=(x,y,z)$ and $z$ denotes the quadrotor's height. To this end, we choose our output\footnote{For the model described in \cite{CosnerICRA24}, the first component of $\by$ has relative degree two whereas the second and third have relative degree one. The theory developed in Sec. \ref{sec:feedback-lin} can be modified to account for such a situation at the expense of additional notation by transforming the output dynamics into a mixed relative degree cascaded system (cf. \cite{AndrewCDC22,CohenARC24}), but a formal presentation of such results is omitted here in the interest of space.} as $\by(\bx) = (z, q_{x}, q_{y})$, where $q_x$ and $q_y$ are components of the quaternion such that $\bq=q_{w} + q_{x}i + q_{y}j + q_{z}k$. Given this output,
we define $\psi(\by(\bx)) = z - z_{\min} - \lambda(2q_{x}^2 + 2q_{y}^2)$ where $\lambda>0$. This constraint ensures that $\psi(\by(\bx))\geq0\implies z\geq z_{\min}$ and requires the quadrotor's orientation to remain level when $z=z_{\min}$. 
Leveraging the constructions in Sec. \ref{sec:feedback-lin}, this leads to the CBF candidate:
\begin{equation*}
    h(\bx) = \psi(\by(\bx)) - \tfrac{1}{2}\|L_{\bf}\by(\bx) - \bk_{1}(\by(\bx))\|^2,
\end{equation*}
where $\bk_1\,:\,\R^3\rightarrow\R^3$ is defined as in \eqref{eq:sontag}.
This CBF is used to construct a safety filter as in \eqref{eq:cbf-qp}, where $\bk_{\mathrm{d}}$ corresponds to commands given via joystick that lift the quadrotor up before lowering it to the ground. Applying this safety filter to the system produces the results in Fig. \ref{fig:master_figure} and Fig. \ref{fig:3d_quad_barriers}, where $z$ remains above $z_{\min}$ and $h$ remains positive for all time.

\begin{figure}
    \centering
    \includegraphics{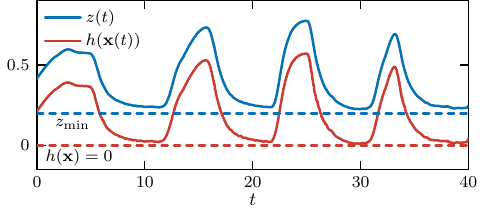}
    \vspace{-0.45cm}
    \caption{Experimental results (cf. Fig. \ref{fig:master_figure}) illustrating the evolution of the quadrotor's height (blue) and CBF (red).}
    \label{fig:3d_quad_barriers}
    \vspace{-0.6cm}
\end{figure}

\section{Conclusions}\label{sec:conclusions}
We presented a framework for synthesizing CBFs using ideas from feedback linearization, which were demonstrated both numerically and experimentally on underactuated robotic systems. Future research directions include characterizing the zero dynamics under CBF-based controllers.

\bibliographystyle{ieeetr}
\bibliography{biblio}

\end{document}